\documentclass[11pt,reqno]{amsart}

\usepackage[margin=1in]{geometry}

\usepackage{amssymb,mathtools}
\usepackage{bm}
\usepackage{cite}
\usepackage{hyperref}
\hypersetup{hidelinks}
\usepackage[capitalize,nameinlink]{cleveref}

\theoremstyle{plain}
\newtheorem{theorem}{Theorem}[section]
\newtheorem{lemma}[theorem]{Lemma}
\newtheorem{proposition}[theorem]{Proposition}
\newtheorem{corollary}[theorem]{Corollary}

\theoremstyle{definition}
\newtheorem{definition}[theorem]{Definition}

\theoremstyle{remark}
\newtheorem{remark}[theorem]{Remark}

\numberwithin{equation}{section}

\AtBeginDocument{\bibliographystyle{plainurl}}
\usepackage{enumitem}
\usepackage[expansion=false]{microtype}
\makeatletter
\newcommand{\printtitleaddresses}{%
  \@setaddresses\global\let\addresses\@empty}
\makeatother

\hypersetup{
  pdftitle={Exact Chained-Word Threshold and Monotone-Path Structure in the Even Root-of-Unity Kaleidoscope Yang--Baxter Algebra},
  pdfauthor={Qihang Wang and Zhiyuan Yao},
  pdfsubject={A common-flag proof, complete sharp-word classification, full radical filtration, and arbitrary-field filtered extension},
  pdfkeywords={Yang--Baxter equation, quantum torus, Jacobson radical, Gabriel quiver, Loewy layers, root of unity}
}

\newcommand{\CC}{\mathbb{C}}
\newcommand{\ZZ}{\mathbb{Z}}
\newcommand{\cK}{\mathcal{K}}
\newcommand{\Id}{\mathrm{I}}
\newcommand{\Span}{\operatorname{span}}

\newcommand{\wtau}{\tau}
\newcommand{\Wword}{\mathcal{W}}
\newcommand{\Dop}{\mathcal{D}}
\newcommand{\fundingstatement}{Zhiyuan Yao acknowledges support from the
Strategic Priority Research Program of the Chinese Academy of Sciences under
Grant No.~XDB1680102 and from the National Natural Science Foundation of China
under Grant No.~12247101.}

\title[Threshold and monotone-path structure]{Exact chained-word threshold and
monotone-path structure in the even root-of-unity Kaleidoscope Yang--Baxter
algebra}
\author{Qihang Wang}
\address[Qihang Wang]{School of Mathematical Sciences, Peking University,
Beijing 100871, China}
\author{Zhiyuan Yao}
\address[Zhiyuan Yao]{Lanzhou Center for Theoretical Physics,\hfil\break
Key Laboratory of Theoretical Physics of Gansu Province,\hfil\break
Key Laboratory of Quantum Theory and Applications of MoE,\hfil\break
Gansu Provincial Research Center for Basic Disciplines of Quantum Physics,
Lanzhou University, Lanzhou 730000, China}
\email{yaozy@lzu.edu.cn}
\thanks{Corresponding author: Zhiyuan Yao.
Qihang Wang's ORCID: \href{https://orcid.org/0009-0003-2504-8756}{0009-0003-2504-8756}.}
\date{}


\begin{document}

\begin{abstract}
Qiu, Guan, and Yu conjectured that in the even-order Kaleidoscope
Yang--Baxter equation every word $\Lambda T^{k_1}\Lambda\cdots T^{k_M}\Lambda$,
built from a square-zero matrix $\Lambda$ and integer powers of a shift $T$
of order $N$, vanishes once $M\geq N/2$.  We prove this for every even
$N\geq4$ and every complex parameter $u$ with $u\neq0$ and $u^N\neq1$, and
show that the bound is sharp.  For $M=N/2-1$ we determine exactly which
words are nonzero; each of them has rank two.  After block reduction, the
factors associated with the intervening powers of $T$ lower a common flag
of length $N/2$ by at least one step.  The spectral idempotents of $T$ then
give a presentation of the algebra generated by $T$ and $\Lambda$ by paths
between levels, in which two paths compose only when the level changes
monotonically.  This algebra is independent of $u$, has dimension
$N^2-N+2$, and is a pullback of two incidence algebras.  Its radical is
generated by $\Lambda$ and has nilpotency index $N/2+1$, and the algebra has
infinite global dimension and infinite representation type.  Replacing
powers of $T$ by arbitrary functions of $T$ gives exact rank formulas and,
for $N\geq6$, rank-one elements in the top radical power.  Finally, four
conditions on a filtered pair of operators over an arbitrary field recover
the threshold and the radical; none of them can be omitted.
\end{abstract}

\maketitle
\printtitleaddresses

\section*{Acknowledgments}

We thank Wen-Jie Qiu, Yi-Cong Yu, and Xi-Wen Guan for a careful reading of the
manuscript and for their comments.  We also thank them for discussions of
their conjecture and of its relation to the Kaleidoscope Yang--Baxter
equation.  We acknowledge Gewu
Intelligence Lab for providing the collaborative research environment in which
this project was developed.  Zhiyuan Yao also acknowledges support from the
Strategic Priority Research Program of the Chinese Academy of Sciences under
Grant No.~XDB1680102 and from the National Natural Science Foundation of China
under Grant No.~12247101.

\clearpage

\section{Introduction}

The Yang--Baxter equation originated in the Bethe wave functions of
one-dimensional many-body problems
\cite{Bethe_1931zt,McGuire_1964so,Gaudin_1967us,Sutherland_1968fr}, as the
condition, isolated by Yang, under which their two-body scattering
factorizes consistently \cite{Yang_1967se,Yang_1968mx}.  Baxter reached it
independently through commuting transfer matrices
\cite{Baxter_1972pf,Baxter_Book}, and in two-dimensional quantum field
theory it became the factorization condition for $S$-matrices
\cite{Zamolodchikov_1979fs,Jimbo_1991it}.  Read as the defining relation of
an algebraic object, it led to quantum groups
\cite{Sklyanin_1982sa,Drinfeld_1985ha,Jimbo_1985ad,Drinfeld_1986qg,Faddeev_1988qo}.
The identity studied here arises from scattering in a kaleidoscope.  Gaudin
interpreted Bethe wave functions as superpositions of waves reflected in the
mirrors of a finite reflection group
\cite{Gaudin2014Bethe,Coxeter_Book_regular_polytopes,Humphreys_Book}.
Earlier Yang--Baxter-type equations for kaleidoscope scattering were found
by Ge, Zhao, and Zhang \cite{GeZhaoZhang1993Kaleidoscope}.  The asymmetric
Bethe ansatz of Jackson et al.\ replaces selected semitransparent mirrors by
perfectly reflecting ones \cite{JacksonEtAl2024Asymmetric}; it has been
applied to the spatially odd sector of two interacting bosons on a ring with
a delta barrier \cite{Olshanii_2025es}.  Yu, Qiu, and Cai study the related
problem of a plane wave scattered by delta barriers that intersect at a
point, through scattering integral equations \cite{Yu_2026qs}.

In Gaudin's kaleidoscope models, consistency of multiple scattering leads
to identities among finite-dimensional operators.  Qiu, Guan, and Yu
established one such identity, the Kaleidoscope Yang--Baxter equation
(KYBE) \cite{QiuGuanYu2026Kaleidoscope}, and proposed a stronger termwise
vanishing conjecture for its generators.  The conjecture asks when every alternating word in a
square-zero generator and powers of a finite-order operator must vanish.
We resolve this question uniformly in the even order and determine the
algebra generated by the same operators.  A common lowering flag explains
the exact word threshold; a spectral decomposition identifies the
multiplication and all radical powers.

Fix an even integer $N=2n\geq4$ and put $\omega=\exp(2\pi i/N)$.  Let $t$ and
$d$ be the shift and clock operators on $\CC^N$: on the standard basis
indexed by $j\in\ZZ/N\ZZ$, $t\bm{e}_j=\bm{e}_{j-1}$ and $d\bm{e}_j=\omega^j\bm{e}_j$, so that
$t^N=d^N=\Id$ and $td=\omega dt$.  This is the Weyl commutation relation of
quantum mechanics \cite{Weyl_Book} in Schwinger's finite-dimensional
realization \cite{Schwinger1960Unitary}, that is, the defining relation of
the quantum torus, or multiplicative Weyl algebra
\cite{McConnellPettit1988}, at the root of unity $\omega$.
For a complex parameter $u$, set
\begin{equation}\label{eq:intro-generators}
  s=(ud-u^{-1}d^{-1})^{-1},\qquad
  T=\begin{pmatrix}t&0\\0&t^{-1}\end{pmatrix},\qquad
  \Lambda=\begin{pmatrix}s&s\\-s&-s\end{pmatrix},
\end{equation}
two operators on $\CC^{2N}$.  The inverse defining $s$ exists exactly when
$u\neq0$ and $u^N\neq1$; we call such $u$ admissible.  Then $T^N=\Id$ and
$\Lambda^2=0$.  Qiu, Guan, and Yu derive the KYBE $[T(\Id+\kappa\Lambda)]^N=\Id$,
an identity in a formal variable $\kappa$, from a spectral argument: for
each value of $\kappa$, the matrix $T(\Id+\kappa\Lambda)$ is diagonalizable,
and its eigenvalues are the $N$th roots of unity.  Expanding the KYBE in powers of $\kappa$ yields
identities in which words in $T$ and $\Lambda$ appear summed over their
exponents \cite[Secs.~4.3--5]{QiuGuanYu2026Kaleidoscope}.  In this
even-order setting they conjecture the termwise statement
\begin{equation}\label{eq:qgy-conjecture}
  \Lambda T^{k_1}\Lambda T^{k_2}\cdots T^{k_M}\Lambda=0
  \quad\text{for all }k_1,\ldots,k_M\in\ZZ
  \quad\text{whenever }M\geq N/2,
\end{equation}
and verify it through $N=10$.  They also state, without proof, that $T$ and
$\Lambda$ generate an algebra $A=S+J$ of Loewy length $N/2$
\cite[Sec.~5]{QiuGuanYu2026Kaleidoscope}.  Here $S$ is spanned by the
powers of $T$, and $J$, the two-sided ideal generated by $\Lambda$, is
proposed as the Jacobson radical.

The distinction between the KYBE and \eqref{eq:qgy-conjecture} is
important.  The spectral identity constrains sums of words; the conjecture
requires each sufficiently long word to vanish separately, for arbitrary
exponents.  For fixed $N$, exponent periodicity reduces this to finitely
many rational-function identities in $u$, but does not supply a proof
uniform in $N$.  Nor does the summed identity identify the multiplication
of the represented algebra.  The termwise statement and the proposed
radical decomposition were left without proof in
\cite[Sec.~5]{QiuGuanYu2026Kaleidoscope}.

Our main theorem gives both the sharp word identity and the algebraic
structure behind it.  It holds throughout the admissible complex domain,
including the unit-circle parametrization of the original conjecture.
We call the level $M=n-1$ immediately below the threshold the sharp layer.
The ideal $J$ is indeed the radical, and its largest nonzero power is
$J^{N/2}$.  In the standard convention, in which the Loewy length is the
nilpotency index of the radical, the Loewy length is therefore $N/2+1$.
Its defining relations include commutativity relations, so the algebra is
not monomial (Remark~\ref{rem:not-monomial}).

\begin{theorem}[Exact chained-word threshold and monotone-path structure]
\label{thm:main}
Let $N=2n\geq4$ be even, let $u$ be admissible, and let $T$ and $\Lambda$ be
as in \eqref{eq:intro-generators}.  Then
\[
  \Lambda T^{k_1}\Lambda\cdots T^{k_M}\Lambda=0
\]
for every $M\geq n$ and every $(k_1,\ldots,k_M)\in\ZZ^M$.  At the sharp
layer $M=n-1$ the word is nonzero if and only if
\begin{equation}\label{eq:sharp-criterion-intro}
  n\nmid j k_j\qquad(1\leq j<n),
\end{equation}
and every nonzero word there has rank two.

Let $A_{N,u}$ be the unital subalgebra of $\operatorname{End}(\CC^{2N})$
generated by $T$ and $\Lambda$, and let $J_{N,u}=A_{N,u}\Lambda A_{N,u}$ be
the two-sided ideal generated by $\Lambda$.  Then $J_{N,u}$ is the Jacobson
radical of $A_{N,u}$, with $J_{N,u}^{\,n}\neq0$ and $J_{N,u}^{\,n+1}=0$.  The
isomorphism type of $A_{N,u}$ is independent of admissible $u$, and
\[
 \dim A_{N,u}=N^2-N+2,
 \qquad
 \dim J_{N,u}^{\,r}=2+4(n-r)(n-r+1)
 \quad(1\leq r\leq n).
\]
Its center is $\CC$, its Cartan matrix has rank $N-1$, and it has infinite
global dimension and infinite representation type.

For $a\in\ZZ/N\ZZ$ let $e_a=\frac1N\sum_{k=0}^{N-1}\omega^{-ak}T^k$, the
spectral idempotent of $T$, and let $\ell(a)=\min\{\bar a,N-\bar a\}$ for the
representative $\bar a\in\{0,\ldots,N-1\}$, the level of $a$.  The $e_a$ are
primitive.  The block $e_aA_{N,u}e_b$ is $\CC e_a$ for $a=b$, zero for
$a\neq b$ of equal level, and one-dimensional otherwise, spanned by an
element $x_{a,b}$.  These can be normalized so that $x_{a,b}x_{b,c}=x_{a,c}$
whenever $\ell(a),\ell(b),\ell(c)$ are strictly increasing or strictly
decreasing, and $x_{a,b}x_{b,c}=0$ whenever the level changes direction.
\end{theorem}

We refer to $A_{N,u}$ as the Kaleidoscope Yang--Baxter algebra of order
$N$; by Theorem~\ref{thm:main} it does not depend on $u$ up to isomorphism.
The proof of Theorem~\ref{thm:main} is completed at the end of
Section~\ref{sec:radical}.  Its two parts use complementary structures.

First, block factorization reduces a chained word with $M$ intervening
powers to a product of $M$ operators $B_k=s(t^k-t^{-k})$, followed by $s$.
It is their simultaneous lowering property that controls arbitrary mixed
products.  On functions of the $N$th roots of unity, put
$x=uz+u^{-1}z^{-1}$ and $\rho=z^n$.  The spaces spanned by $x^r$ and
$\rho x^r$ for $r<m$, with $m=0,\ldots,n$, form a flag of length $n$.  Each of these operators,
for every $k$, lowers the same flag by at least one step.  The flag length
therefore forces vanishing, while the induced maps between its successive
quotients give the exact sharp-layer criterion.

Second, the spectral idempotents of $T$ split the generated algebra into
blocks.  A finite Fourier expansion of $s$ determines these blocks and
their rank-one image and row directions.  Their pairings vanish when a
path reverses its level direction and are nonzero along monotone paths.
Normalizing these pairings removes $u$ from the multiplication table.
Thus parameter independence follows from an explicit algebra isomorphism,
and the radical filtration is read off from level distances.
The upward and downward paths span the incidence algebras of the level
order and of its opposite \cite{Rota1964}, and $A_{N,u}$ is their pullback
over the common semisimple quotient, an instance of the construction in
\cite[Sec.~3]{XiXu2013}.  The point here is the explicit identification
of this structure in the KYBE representation, including the normalization
that removes $u$.

The paper contains two further results.
Theorem~\ref{thm:functional-sharp} replaces powers of $T$ by arbitrary
elements of its cyclic functional calculus.  The two parity sectors then
give separate product scalars, which determine whether the sharp word has
rank zero, one, or two; rank one occurs for every even $N\geq6$.
Theorem~\ref{thm:axiomatic-main} works over an arbitrary field with an
exhaustive flag of unrestricted successive dimensions.  Four hypotheses
on preservation, lowering, survival at the penultimate level, and
semisimplicity recover the threshold and radical identification.
The witnesses in Theorem~\ref{thm:axiom-independence} show that none of
these hypotheses can be omitted.  The explicit root-of-unity presentation,
by contrast, uses the additional Fourier structure of the cyclic model.

For the KYBE itself, the outcome is as follows.  In the expansion of
$[T(\Id+\kappa\Lambda)]^N$, every term containing at least $n+1$ copies of
$\Lambda$ vanishes individually, so the coefficients of
$\kappa^{n+1},\ldots,\kappa^N$ vanish termwise.  The coefficients of
$\kappa,\ldots,\kappa^n$ are sums of words that need not vanish
individually; for them the spectral identity constrains only the sums.

Sections~\ref{sec:setup}--\ref{sec:sharpness} establish the word reduction,
common flag, sharp classification, and functional-calculus formulas.
Section~\ref{sec:radical} determines the algebra and its structural
consequences.  Section~\ref{sec:axiomatic-extension} isolates the
arbitrary-field mechanism and proves independence of its four hypotheses.

\section{Operators, parameter domain, and threshold}\label{sec:setup}

Fix an even integer $N=2n\geq4$ and put
\[
  \omega=\exp(2\pi i/N).
\]
Index the standard basis of $\CC^N$ by $j\in\ZZ/N\ZZ$ and define
\begin{equation}\label{eq:shift-clock}
  t \bm{e}_j=\bm{e}_{j-1},\qquad d \bm{e}_j=\omega^j \bm{e}_j.
\end{equation}
These are the shift and clock operators of the finite Weyl representation
\cite{Schwinger1960Unitary}.
Then
\begin{equation}\label{eq:torus-relations}
  t^N=d^N=\Id,\qquad td=\omega dt.
\end{equation}
For $u\in\CC^\times$, set
\begin{equation}\label{eq:generators}
  y=ud-u^{-1}d^{-1},\qquad s=y^{-1},\qquad
  T=\begin{pmatrix}t&0\\0&t^{-1}\end{pmatrix},\qquad
  \Lambda=\begin{pmatrix}s&s\\-s&-s\end{pmatrix},
\end{equation}
whenever $y$ is invertible. These are the conventions used throughout;
matrix products are composed from right to left. With $u=e^{i\theta}$ the
diagonal entries of $y$ are $2i\sin(\theta+2\pi j/N)$.  The $s$ used here
is therefore $1/(2i)$ times the matrix displayed in
\cite[Sec.~2]{QiuGuanYu2026Kaleidoscope}, whose entries are
$[\sin(\theta+2\pi j/N)]^{-1}$; the factor $2i$ is absorbed there into
the coupling $c/(2ik_0)$ that multiplies $\Lambda$.  Rescaling $\Lambda$ by
a nonzero constant changes nothing below.  A word with $M+1$ factors of
$\Lambda$ is multiplied by $(2i)^{-(M+1)}$, so its vanishing, its rank, and
the algebra generated by $T$ and $\Lambda$ are unaffected.

\begin{lemma}[Exact exceptional set]\label{lem:parameter-domain}
The operator $s$ in \eqref{eq:generators} exists if and only if
\begin{equation}\label{eq:admissible-domain}
  u\neq0\qquad\text{and}\qquad u^N\neq1.
\end{equation}
\end{lemma}

\begin{proof}
The diagonal entry of $y$ at site $j$ is
$y_j=u\omega^j-u^{-1}\omega^{-j}$.  For $u\neq0$, it vanishes precisely
when $u^2=\omega^{-2j}$.  Such an equality implies
$u^N=(u^2)^n=1$.  Conversely, if $u^N=1$, write $u=\omega^\ell$ and choose
$j\equiv-\ell\pmod n$; then $u^2=\omega^{-2j}$.  The value $u=0$ is
excluded already by the occurrence of $u^{-1}$.
\end{proof}

For fixed $N$, $M$, and exponent tuple, each entry of a chained word is a
rational function of $u$.  Thus vanishing on the admissible unit circle
already implies vanishing throughout \eqref{eq:admissible-domain} by the
identity principle for rational functions.  Our proof works directly on
this domain and also determines the nonzero words and the algebra structure
at every admissible parameter.

For $M\in\ZZ_{\geq0}$ and $\boldsymbol{k}=(k_1,\ldots,k_M)\in\ZZ^M$, write
\begin{equation}\label{eq:word-definition}
  \Wword_{N,M}(\boldsymbol{k};u)
  =\Lambda T^{k_1}\Lambda\cdots T^{k_M}\Lambda.
\end{equation}
Thus $M$ counts the intervening powers of $T$, and the word contains $M+1$
copies of $\Lambda$.  For admissible $u$, define
\begin{equation}\label{eq:tau-definition}
  \wtau_N(u)=\min\bigl\{m\in\ZZ_{\geq0}:
  \Wword_{N,M}(\boldsymbol{k};u)=0
  \text{ for every }M\geq m\text{ and every }\boldsymbol{k}\in\ZZ^M
  \bigr\}.
\end{equation}

The next lemma is the elementary flag-lowering criterion for nilpotence;
for the general relation between flags and nilpotent matrix semigroups,
see \cite[Sec.~3]{KudryavtsevaMazorchuk2008}.

\begin{lemma}[Common-flag threshold principle]\label{lem:abstract-threshold}
Let $V$ be a vector space with a finite flag
\[
  0=F_0\subseteq F_1\subseteq\cdots\subseteq F_n=V,
\]
let $S_0\in\operatorname{End}(V)$, and let
$\mathcal B\subseteq\operatorname{End}(V)$ be a set of operators with
\[
  BF_m\subseteq F_{m-1}
  \qquad(1\leq m\leq n,\ B\in\mathcal B).
\]
Then every product $B_1\cdots B_MS_0$ with $B_1,\ldots,B_M\in\mathcal B$ is
zero for $M\geq n$.  If, in addition, there are $\bm v\in V$ and
$B_*\in\mathcal B$ such that
\[
  B_*^{\,n-1}S_0\bm v\neq0,
\]
then $n$ is the least integer $m$ for which all such products with
$M\geq m$ vanish.
\end{lemma}

\begin{proof}
The map $S_0$ sends $V$ into $F_n$.  Reading a product from right to left,
each of the next $n$ lowering maps drops the image by one flag level, so it
reaches $F_0=0$.  This also annihilates every longer product.  The displayed
nonzero vector excludes vanishing at level $n-1$, proving exactness.
\end{proof}

In this notation, Theorem~\ref{thm:main} asserts
\begin{equation}\label{eq:main-threshold}
  \wtau_N(u)=\frac{N}{2}
\end{equation}
throughout the admissible parameter domain.

\section{Direct block reduction}\label{sec:reduction}

Let
\[
  \mathbf q=\begin{pmatrix}1\\-1\end{pmatrix},
  \qquad \mathbf r=\begin{pmatrix}1&1\end{pmatrix},
\]
where scalar entries denote identity operators on $\CC^N$.  Then
\begin{equation}\label{eq:rank-one-identities}
  \Lambda=\mathbf q\,s\,\mathbf r,
  \qquad
  \mathbf r T^k\mathbf q=t^k-t^{-k}=:\Delta_k.
\end{equation}
Consequently, direct multiplication in the original representation gives
\begin{equation}\label{eq:rank-one-word}
  \Wword_{N,M}(\boldsymbol{k};u)
  =\mathbf q\,\Phi_{N,M}(\boldsymbol{k};u)\,\mathbf r,
  \qquad
  \Phi_{N,M}=s\Delta_{k_1}s\cdots\Delta_{k_M}s.
\end{equation}
The order of the factors is unchanged, and the word is zero if and only if
$\Phi_{N,M}$ is zero.  Indeed,
its two block rows are $(\Phi_{N,M},\Phi_{N,M})$ and
$(-\Phi_{N,M},-\Phi_{N,M})$.

Because $T^N=\Id$, each exponent may be reduced modulo $N$.  We use integer
exponents below, with all identities depending only on their residue classes.

\section{The cyclic model and its common flag}\label{sec:flag}

Fix an admissible $u$.  Work in the cyclic algebra
\begin{equation}\label{eq:cyclic-algebra}
  R=\CC[z,z^{-1}]/(z^N-1).
\end{equation}
Evaluation at the $N$ roots $z=\omega^j$ identifies $R$ with $\CC^N$.
Under this identification, $d$ is multiplication by $z$ and
\begin{equation}\label{eq:t-action}
  (tf)(z)=f(\omega z).
\end{equation}
Put
\begin{equation}\label{eq:axrho}
  x=uz+u^{-1}z^{-1},\qquad y=uz-u^{-1}z^{-1},\qquad s=y^{-1},\qquad
  \rho=z^n.
\end{equation}
Lemma~\ref{lem:parameter-domain} says exactly that $y$ is a unit in $R$.
We also have
\begin{equation}\label{eq:a-square}
  y^2=x^2-4,
  \qquad \rho(\omega^j)=(-1)^j.
\end{equation}

\begin{lemma}[Parity basis]\label{lem:parity-basis}
The $N$ elements
\begin{equation}\label{eq:parity-basis}
  1,x,\ldots,x^{n-1},\rho,\rho x,\ldots,\rho x^{n-1}
\end{equation}
form a basis of $R$.
\end{lemma}

\begin{proof}
Suppose $P(x)+\rho Q(x)=0$ in $R$, where
$\deg P,\deg Q<n$.  At the $n$ even-indexed roots, $P+Q$ vanishes; at the
$n$ odd-indexed roots, $P-Q$ vanishes.

Within either parity class, the corresponding values of $x$ are distinct.
Indeed, if $j\not\equiv\ell\pmod N$ have the same parity and
$x(\omega^j)=x(\omega^\ell)$, then rearrangement gives
\[
  u^2=\omega^{-(j+\ell)}.
\]
Since $j+\ell$ is even, this is one of the pole conditions in
Lemma~\ref{lem:parameter-domain}, contradicting admissibility.  Thus $P+Q$
and $P-Q$, each of degree below $n$, vanish on $n$ distinct points and are
both zero.  Hence $P=Q=0$.  There are $N$ independent elements in the
$N$-dimensional algebra $R$.
\end{proof}

For $0\leq m\leq n$, define
\begin{equation}\label{eq:flag}
  \cK_m=\Span_{\CC}\{x^r,\rho x^r:0\leq r<m\}.
\end{equation}
Then $\cK_0=0$, $\dim\cK_m=2m$, and $\cK_n=R$.

For $k\in\ZZ$, set
\begin{equation}\label{eq:ck-ek}
  c_k=\frac{\omega^k+\omega^{-k}}{2},\qquad
  \eta_k=\frac{\omega^k-\omega^{-k}}{2}.
\end{equation}
Direct substitution yields
\begin{equation}\label{eq:shifted-x-rho}
  x(\omega^{\pm k}z)=c_kx\pm \eta_ky,
  \qquad
  \rho(\omega^{\pm k}z)=(-1)^k\rho.
\end{equation}

For $P\in\CC[X]$, introduce the divided-difference polynomial
\begin{equation}\label{eq:divided-difference}
  (\Dop_kP)(X)=
  \left.
  \frac{P(c_kX+\eta_kY)-P(c_kX-\eta_kY)}{Y}
  \right|_{Y^2=X^2-4}.
\end{equation}
The numerator is odd in $Y$, so it is divisible by $Y$ and the quotient is
a polynomial in $X$ and $Y^2$.  The substitution in
\eqref{eq:divided-difference} is therefore legitimate.  Expansion of a
monomial gives
\begin{equation}\label{eq:divided-expansion}
  \Dop_kX^r=
  2\sum_{\substack{1\leq j\leq r\\j\ \mathrm{odd}}}
  \binom{r}{j}c_k^{r-j}\eta_k^j
  X^{r-j}(X^2-4)^{(j-1)/2}.
\end{equation}
Thus $\deg(\Dop_kP)\leq\deg P-1$, and constants are mapped to zero.

\begin{proposition}[Common lowering rule]\label{prop:lowering}
For
\begin{equation}\label{eq:Bk}
  B_k=s(t^k-t^{-k}),
\end{equation}
one has
\begin{equation}\label{eq:Bk-formula}
  B_k\bigl(P(x)+\rho Q(x)\bigr)
  =\Dop_kP(x)+(-1)^k\rho\Dop_kQ(x).
\end{equation}
Consequently,
\begin{equation}\label{eq:flag-lowering}
  B_k\cK_m\subseteq\cK_{m-1}
  \qquad(1\leq m\leq n, k\in\ZZ).
\end{equation}
\end{proposition}

\begin{proof}
Let $f=P(x)+\rho Q(x)$.  From \eqref{eq:t-action} and
\eqref{eq:shifted-x-rho},
\[
  \bigl(t^{\pm k}P(x)\bigr)(z)=P\bigl(c_kx\pm\eta_ky\bigr),
\]
and likewise for $Q$, together with the sign
$\rho(\omega^{\pm k}z)=(-1)^k\rho$.  Therefore
\[
  (t^k-t^{-k})P(x)
  =P(c_kx+\eta_ky)-P(c_kx-\eta_ky).
\]
Since $y^2=x^2-4$ in $R$, multiplying by $s=y^{-1}$ produces $\Dop_kP(x)$ by
\eqref{eq:divided-difference}.  The same computation on $\rho Q(x)$ supplies
the extra factor $(-1)^k$, which is \eqref{eq:Bk-formula}.  The degree bound
following \eqref{eq:divided-expansion} gives \eqref{eq:flag-lowering}.
\end{proof}

\begin{proposition}[Uniform vanishing]\label{prop:vanishing}
For every $M\geq n$ and every $(k_1,\ldots,k_M)\in\ZZ^M$,
\[
  \Wword_{N,M}(\boldsymbol{k};u)=0.
\]
\end{proposition}

\begin{proof}
Equation \eqref{eq:rank-one-word} may be written as
\begin{equation}\label{eq:F-as-B}
  \Phi_{N,M}=B_{k_1}B_{k_2}\cdots B_{k_M}s.
\end{equation}
Apply Lemma~\ref{lem:abstract-threshold} with $V=R$, $F_m=\cK_m$, $S_0=s$,
and $\mathcal B=\{B_k:k\in\ZZ\}$.  Proposition~\ref{prop:lowering}
supplies its lowering hypothesis, so $\Phi_{N,M}=0$ for $M\geq n$.
Equation~\eqref{eq:rank-one-word} gives the conjectured conclusion.
\end{proof}

\section{Triangular products and the complete sharp layer}
\label{sec:sharpness}

Fix $k\in\ZZ$ and $r\geq1$.  The coefficient of $X^{r-1}$ in
\eqref{eq:divided-expansion} is
\begin{align}
  \lambda_{r,k}
  &=2\sum_{\substack{1\leq j\leq r\\j\ \mathrm{odd}}}
    \binom{r}{j}c_k^{r-j}\eta_k^j \notag\\
  &=(c_k+\eta_k)^r-(c_k-\eta_k)^r
    =\omega^{kr}-\omega^{-kr}.\label{eq:lambda-r}
\end{align}
All remaining terms have degree at most $r-2$.  Under composition these
coefficients multiply, even when some of them vanish; this is the content
of the next lemma.

\begin{lemma}[Associated graded product]\label{lem:graded-product}
Let $r,M\in\ZZ_{\geq0}$ and $k_1,\ldots,k_M\in\ZZ$.  If $M\leq r$, then
\begin{align}
  \Dop_{k_1}\cdots\Dop_{k_M}X^r
  &={L}_{r,M}(\boldsymbol{k})X^{r-M}
    +\text{terms of degree at most $r-M-1$},\label{eq:graded-product}\\
  {L}_{r,M}(\boldsymbol{k})
  &:=\prod_{j=1}^{M}
  \bigl(\omega^{k_j(r-M+j)}-\omega^{-k_j(r-M+j)}\bigr).
  \label{eq:graded-scalar}
\end{align}
If $M>r$, the same operator product annihilates $X^r$.
\end{lemma}

\begin{proof}
We argue by induction on $M$; the case $M=0$ is immediate.  For $M\geq1$,
recall that the rightmost operator acts first.  Applying the induction hypothesis to
$\Dop_{k_2}\cdots\Dop_{k_M}X^r$ gives a polynomial of degree at most
$r-M+1$, whose coefficient of $X^{r-M+1}$ is the product of the factors
indexed by $2\leq j\leq M$.  By \eqref{eq:lambda-r}, the coefficient of
$X^{r-M}$ in $\Dop_{k_1}X^{r-M+1}$ is $\lambda_{r-M+1,k_1}$.
Multiplying these coefficients supplies the missing factor, whether or
not it is zero.  Lower-degree terms cannot contribute to degree $r-M$.
The assertion for $r<M$ follows by repeated degree lowering.
\end{proof}

Repeated use of \eqref{eq:Bk-formula} also gives the two-sector identity
\begin{equation}\label{eq:mixed-B-product}
 B_{k_1}\cdots B_{k_M}\bigl(P(x)+\rho Q(x)\bigr)
 =\Dop_{k_1}\cdots\Dop_{k_M}P(x)
 +(-1)^{k_1+\cdots+k_M}\rho
  \Dop_{k_1}\cdots\Dop_{k_M}Q(x).
\end{equation}

\begin{proposition}[Complete sharp-layer classification]
\label{prop:sharp-classification}
At $M=n-1$, put
\begin{equation}\label{eq:mixed-sharp-scalar}
  C_N(\boldsymbol{k})
  :=\prod_{j=1}^{n-1}
  \bigl(\omega^{j k_j}-\omega^{-j k_j}\bigr).
\end{equation}
Then
\begin{equation}\label{eq:sharp-equivalences}
 \Wword_{N,n-1}(\boldsymbol{k};u)\neq0
 \quad\Longleftrightarrow\quad
 C_N(\boldsymbol{k})\neq0
 \quad\Longleftrightarrow\quad
 n\nmid j k_j\quad(1\leq j<n).
\end{equation}
Every nonzero sharp word has rank two.  The vector
\begin{equation}\label{eq:CN}
  \bm{v}_N=yx^{n-1}\in R
\end{equation}
is a certificate with
\begin{equation}\label{eq:mixed-sharp-certificate}
  B_{k_1}\cdots B_{k_{n-1}}s\,\bm{v}_N
  =C_N(\boldsymbol{k})\,1.
\end{equation}
Finally, among the $N^{n-1}$ exponent tuples modulo $N=2n$, exactly
\begin{equation}\label{eq:sharp-count}
  2^{n-1}\prod_{j=1}^{n-1}\bigl(n-\gcd(j,n)\bigr)
\end{equation}
give nonzero sharp words.
\end{proposition}

\begin{proof}
After $n-1$ lowering maps, every input in $\cK_{n-1}$ is annihilated.
Equations \eqref{eq:graded-product}--\eqref{eq:mixed-B-product} show that
on the two remaining basis vectors the product acts by
\begin{align*}
 B_{k_1}\cdots B_{k_{n-1}}x^{n-1}
 &=C_N(\boldsymbol{k})\,1,\\
 B_{k_1}\cdots B_{k_{n-1}}\rho x^{n-1}
 &=(-1)^{k_1+\cdots+k_{n-1}}C_N(\boldsymbol{k})\,\rho.
\end{align*}
Thus the product is zero exactly when $C_N(\boldsymbol{k})=0$; otherwise
its image is $\cK_1=\Span\{1,\rho\}$ and its rank is two.  The rightmost
$s$ in \eqref{eq:F-as-B} is invertible.  Moreover, $\mathbf r$ is
surjective and $\mathbf q$ is injective, so \eqref{eq:rank-one-word}
preserves the rank of the normalized block.  This proves the first two
claims and \eqref{eq:mixed-sharp-certificate}.

The $j$th factor in \eqref{eq:mixed-sharp-scalar} vanishes exactly when
$\omega^{2jk_j}=1$, equivalently $n\mid j k_j$.  Put
$g_j=\gcd(j,n)$.  There are $2g_j$ residues $k_j$ modulo $2n$ satisfying
$n\mid j k_j$, hence $2(n-g_j)$ admissible residues.  Multiplication of
these independent coordinate counts gives \eqref{eq:sharp-count}.
\end{proof}

The same triangular calculation extends from monomials to arbitrary functions
of $T$.  Let
\begin{equation}\label{eq:cyclic-functional-calculus}
  \mathcal C_N=\CC[X,X^{-1}]/(X^N-1).
\end{equation}
Every $f\in\mathcal C_N$ has a unique representative
$f(X)=\sum_{k=0}^{N-1}\widehat f(k)X^k$.  For such an $f$, define
\begin{equation}\label{eq:Bf}
 B_f=s\bigl(f(t)-f(t^{-1})\bigr)
\end{equation}
and
\begin{equation}\label{eq:Df-pm}
 \Dop_f^+=\sum_{k=0}^{N-1}\widehat f(k)\,\Dop_k,
 \qquad
 \Dop_f^-=\sum_{k=0}^{N-1}(-1)^k\widehat f(k)\,\Dop_k.
\end{equation}
For $\boldsymbol f=(f_1,\ldots,f_M)\in\mathcal C_N^M$, put
\begin{equation}\label{eq:functional-word}
 \Wword_{N,M}(\boldsymbol f;u)
 =\Lambda f_1(T)\Lambda\cdots f_M(T)\Lambda.
\end{equation}

\begin{theorem}[Functional sharp layer]\label{thm:functional-sharp}
For every $M\geq n$ and every $\boldsymbol f\in\mathcal C_N^M$,
\begin{equation}\label{eq:functional-vanishing}
  \Wword_{N,M}(\boldsymbol f;u)=0.
\end{equation}
At $M=n-1$, define
\begin{align}
 C_N^+(\boldsymbol f)
 &=\prod_{j=1}^{n-1}
   \bigl(f_j(\omega^j)-f_j(\omega^{-j})\bigr),\label{eq:functional-plus}\\
 C_N^-(\boldsymbol f)
 &=\prod_{j=1}^{n-1}
   \bigl(f_j(-\omega^j)-f_j(-\omega^{-j})\bigr).
   \label{eq:functional-minus}
\end{align}
Then the exact rank is
\begin{equation}\label{eq:functional-rank}
 \operatorname{rank}\Wword_{N,n-1}(\boldsymbol f;u)
 =\boldsymbol 1_{C_N^+(\boldsymbol f)\neq0}
  +\boldsymbol 1_{C_N^-(\boldsymbol f)\neq0}.
\end{equation}
In particular, the functional sharp word is zero if and only if both product
scalars vanish.  The two vectors
\begin{equation}\label{eq:functional-certificates}
 \bm{v}_N^+=yx^{n-1},\qquad \bm{v}_N^-=y\rho x^{n-1}
\end{equation}
satisfy
\begin{align}
 B_{f_1}\cdots B_{f_{n-1}}s\,\bm{v}_N^+
 &=C_N^+(\boldsymbol f)\,1,\label{eq:functional-certificate-plus}\\
 B_{f_1}\cdots B_{f_{n-1}}s\,\bm{v}_N^-
 &=C_N^-(\boldsymbol f)\,\rho.
 \label{eq:functional-certificate-minus}
\end{align}
For $N=4$ one has $C_4^-=-C_4^+$, so the rank is zero or two.  For every
even $N\geq6$, rank one occurs; consequently, the top radical power $J^n$
contains a rank-one operator.
\end{theorem}

\begin{proof}
Since $N$ is even, the signs $(-1)^k$ in \eqref{eq:Df-pm} are well defined
modulo $N$.  Linearity of \eqref{eq:Bk-formula} gives
\begin{equation}\label{eq:Bf-formula}
 B_f\bigl(P(x)+\rho Q(x)\bigr)
 =\Dop_f^+P(x)+\rho\Dop_f^-Q(x).
\end{equation}
Both polynomial operators lower degree by at least one.  Moreover, their
coefficients of $X^{r-1}$ on the input $X^r$ are respectively
\begin{align}
 \lambda_r^+(f)&=f(\omega^r)-f(\omega^{-r}),\label{eq:functional-leading-plus}\\
 \lambda_r^-(f)&=f(-\omega^r)-f(-\omega^{-r}).
 \label{eq:functional-leading-minus}
\end{align}
Indeed, these identities follow by summing \eqref{eq:lambda-r} with weights
$\widehat f(k)$ and $(-1)^k\widehat f(k)$.

The block reduction extends verbatim because
\begin{equation}\label{eq:functional-rank-one-reduction}
 \mathbf r f(T)\mathbf q=f(t)-f(t^{-1}),\qquad
 \Wword_{N,M}(\boldsymbol f;u)
 =\mathbf q B_{f_1}\cdots B_{f_M}s\,\mathbf r.
\end{equation}
The common flag therefore proves \eqref{eq:functional-vanishing}.  At
$M=n-1$, the product annihilates $\cK_{n-1}$.  On the two complementary
basis vectors $x^{n-1}$ and $\rho x^{n-1}$, repeated use of
\eqref{eq:functional-leading-plus}--\eqref{eq:functional-leading-minus}
gives the two scalars \eqref{eq:functional-plus}--\eqref{eq:functional-minus}.
Their target vectors $1$ and $\rho$ are independent.  The rightmost $s$ is
invertible, while $\mathbf r$ is surjective and $\mathbf q$ is injective.
This proves the rank formula and the two certificate identities.

For $N=4$ we have $n=2$ and $\omega=i$, so
$C_4^-(f_1)=f_1(-i)-f_1(i)=-C_4^+(f_1)$.  It remains to exhibit rank one
for $N\geq6$.  Set
\[
 \alpha=\omega-\omega^{-1},\qquad
 \beta=\omega^2-\omega^{-2}.
\]
Both are nonzero, since $\omega^2\neq1$ and $\omega^4\neq1$ for $N\geq6$.
Take
\begin{equation}\label{eq:rank-one-functional-example}
 f_1(X)=X-\frac{\alpha}{\beta}X^2,\qquad
 f_j(X)=X\quad(2\leq j<n).
\end{equation}
The $j=1$ factor of $C_N^+$ is zero, whereas the corresponding factor of
$C_N^-$ is $-2\alpha$.  Every remaining factor is nonzero because
$1\leq j<n$.  Hence $C_N^+=0$ and $C_N^-\neq0$, so the word has rank one.
It contains $n$ copies of $\Lambda$ and is a product of $n$ elements of $J$;
therefore it belongs to $J^n$.
\end{proof}

Constant exponent tuples are an immediate specialization of the sharp-layer
classification.

\begin{corollary}[Constant exponents]\label{cor:constant-classification}
For $\boldsymbol{k}=(k,\ldots,k)$ at $M=n-1$, the word is nonzero if and
only if $\gcd(k,n)=1$.  In that case
\begin{equation}\label{eq:sharp-block}
  C_N(k,\ldots,k)
  =\prod_{r=1}^{n-1}(\omega^{kr}-\omega^{-kr})
  =(-1)^{n-1}n\,\omega^{-kn(n-1)/2}.
\end{equation}
Consequently exactly $2\varphi(n)$ residue classes of $k$ modulo $N$
give constant-exponent witnesses.
\end{corollary}

\begin{proof}
The conditions $n\nmid jk$ for every $1\leq j<n$ hold exactly when $k$ is
a unit modulo $n$.  If so, put $\zeta=\omega^2$.  Multiplication by $k$
permutes the nonzero residue classes modulo $n$, and therefore
\begin{align*}
 C_N(k,\ldots,k)
 &=\omega^{-k\sum_{r=1}^{n-1}r}
   \prod_{r=1}^{n-1}(\zeta^{kr}-1)\\
 &=\omega^{-kn(n-1)/2}(-1)^{n-1}n,
\end{align*}
using $\prod_{r=1}^{n-1}(1-\zeta^r)=n$.  Every unit class modulo $n$ has
two lifts modulo $2n$, which gives the count.
\end{proof}

For the all-ones tuple, \eqref{eq:sharp-block} specializes to
\begin{equation}\label{eq:certificate-identity}
  B_1^{n-1}s\,\bm{v}_N=n i^{\,n-1}\,1\neq0.
\end{equation}
The output is independent of $u$, so the word remains nonzero at every
admissible specialization.  Proposition~\ref{prop:vanishing} and
\eqref{eq:certificate-identity} prove the exact-threshold part of
Theorem~\ref{thm:main}.

\begin{remark}[Smallest case]
For $N=4$, one has $n=2$.  Proposition~\ref{prop:vanishing} starts at
$M=2$, while the sharp word is $\Lambda T\Lambda$ at $M=1$ and the
certificate scalar is $C_4(1)=\omega-\omega^{-1}=2i$.  More generally,
the sharp word $\Lambda T^{k_1}\Lambda$ is nonzero exactly when $k_1$ is
odd.  Thus the threshold is sharp already for the smallest order $N=4$.
\end{remark}

\section{The represented algebra and its radical filtration}
\label{sec:radical}

The common flag controls long words, but does not by itself determine the
algebra.  We now use the eigenspaces of $T$ to identify every block of
the represented algebra and its multiplication.  The resulting
presentation makes radical powers and projective layers explicit.

The structural results are stated in the language of finite-dimensional
algebras and quiver representations.  Gabriel showed that a connected quiver
has only finitely many isomorphism classes of indecomposable representations
exactly when its underlying graph is a Dynkin diagram of type $\mathbb A$,
$\mathbb D$, or $\mathbb E$ \cite{Gabriel1972Unzerlegbare}; Bernstein,
Gelfand, and Ponomarev reproved this with Coxeter functors
\cite{Bernstein_1973cf}.  The classification of quadruples of subspaces
\cite{Gelfand_1972po} settles a basic case beyond the Dynkin diagrams, the
Euclidean diagram $\widetilde{\mathbb D}_4$.  Over an algebraically closed field,
every finite-dimensional algebra is of tame or of wild representation type,
and never of both \cite{Drozd_1980ta}.  For the roles of integral quadratic
forms and Auslander--Reiten sequences in representation theory, see
\cite{Ringel_Book} and \cite{Gabriel_1980ar}, respectively.  Representation
type is preserved by stable equivalence \cite{Krause_1997se}.  We use the standard
finite-dimensional algebra and module conventions of
\cite{AssemSimsonSkowronski2006,Auslander_Book}; the Gabriel quiver records
the first radical layer between primitive idempotents
\cite{Gabriel1972Unzerlegbare}.

Let
\[
 A=A_{N,u}\subseteq\operatorname{End}(\CC^{2N})
\]
be the unital algebra generated by $T$ and $\Lambda$, and put
\begin{equation}\label{eq:radical-ideal}
 J=A\Lambda A.
\end{equation}
For $a\in\ZZ/N\ZZ$, write $\bar a\in\{0,\ldots,N-1\}$ for its standard
representative and define its level by
\begin{equation}\label{eq:level-function}
 \ell(a)=\min\{\bar a,N-\bar a\}\in\{0,\ldots,n\}.
\end{equation}
Thus the level sizes are $1,2,\ldots,2,1$.

Let $p_a\in\operatorname{End}(R)$ be projection onto the monomial line
$\CC z^a$.  Since $tz^a=\omega^a z^a$, the spectral idempotent
\begin{equation}\label{eq:spectral-idempotent}
 e_a=\frac1N\sum_{k=0}^{N-1}\omega^{-ak}T^k
\end{equation}
projects $R\oplus R$ onto
\begin{equation}\label{eq:spectral-space}
 V_a=\CC(z^a,0)\oplus\CC(0,z^{-a}).
\end{equation}
The $e_a$ are pairwise orthogonal and sum to the identity.

The decisive additional identity is a finite Fourier expansion of $s$.
In the cyclic algebra $R$, one has
\begin{equation}\label{eq:s-fourier-expansion}
 s=\frac1{u^N-1}\sum_{h=0}^{n-1}u^{2h+1}z^{2h+1}.
\end{equation}
Indeed,
\begin{align*}
 (uz-u^{-1}z^{-1})\sum_{h=0}^{n-1}u^{2h+1}z^{2h+1}
 &=\sum_{h=0}^{n-1}u^{2h+2}z^{2h+2}
   -\sum_{h=0}^{n-1}u^{2h}z^{2h}\\
 &=u^Nz^N-1=u^N-1.
\end{align*}
For $q\in\ZZ/N\ZZ$, set $\gamma_q=0$ when $q$ is even.
If $q$ is odd, let $[q]\in\{1,3,\ldots,N-1\}$ be its unique odd
representative and set
\begin{equation}\label{eq:cq-definition}
 \gamma_q=\frac{u^{[q]}}{u^N-1}.
\end{equation}
Then
\begin{equation}\label{eq:s-matrix-coefficient}
 p_a s p_b=\gamma_{a-b}E_{a,b},
\end{equation}
where $E_{a,b}$ sends $z^b$ to $z^a$.  In particular, for odd $q$,
\begin{equation}\label{eq:cq-product}
 \gamma_q\gamma_{-q}=\frac{u^N}{(u^N-1)^2}.
\end{equation}

\begin{theorem}[Monotone-path presentation]
\label{thm:monotone-algebra}
For every admissible $u$, the idempotents $e_a$ are primitive.  For each
ordered pair $(a,b)$ with $\ell(a)\neq\ell(b)$ there is a nonzero element
$x_{a,b}\in e_aAe_b$ such that
\begin{equation}\label{eq:peirce-blocks}
 e_aAe_b=
 \begin{cases}
   \CC e_a,&a=b,\\
   0,&a\neq b\text{ and }\ell(a)=\ell(b),\\
   \CC x_{a,b},&\ell(a)\neq\ell(b).
 \end{cases}
\end{equation}
The elements may be normalized so that
\begin{equation}\label{eq:monotone-multiplication}
 x_{a,b}x_{b,c}=x_{a,c}
\end{equation}
when $\ell(a),\ell(b),\ell(c)$ are strictly increasing or strictly
decreasing, whereas the product is zero when the level direction changes.
Products with mismatched middle idempotents are zero.

Consequently, the isomorphism type of $A_{N,u}$ is independent of admissible
$u$.  Its Gabriel quiver \cite{Gabriel1972Unzerlegbare} has the residues
$a\in\ZZ/N\ZZ$ as vertices and one arrow in each direction between every
pair of vertices in consecutive levels.  Every path that changes level
direction is zero, and all monotone paths with the same endpoints agree
after the normalization above.
\end{theorem}

\begin{proof}
In the ordered bases of $V_b$ and $V_a$ from
\eqref{eq:spectral-space}, equations
\eqref{eq:generators} and \eqref{eq:s-matrix-coefficient} give
\begin{equation}\label{eq:lambda-peirce-block}
 e_a\Lambda e_b=
 \begin{pmatrix}
   \gamma_{a-b}&\gamma_{a+b}\\
   -\gamma_{-a-b}&-\gamma_{-a+b}
 \end{pmatrix}.
\end{equation}
If $a$ and $b$ have the same parity, this matrix is zero.  If they have
opposite parity, every entry is nonzero and its determinant vanishes by
\eqref{eq:cq-product}.  Thus every nonzero one-$\Lambda$ Peirce block has
rank one.  Parity is unchanged by $a\mapsto-a$, so it is also the parity of
$\ell(a)$.

The rank-one directions make the multiplication transparent.  Write an
interior residue as $a=\varepsilon i$, where $1\leq i<n$ and
$\varepsilon\in\{1,-1\}$.  In the basis of $V_a$, maps entering level $i$
from a higher or lower level have respective column lines
\begin{equation}\label{eq:local-column-lines}
 \bm{v}_a^{H}=\binom{1}{-u^{-2\varepsilon i}},
 \qquad
 \bm{v}_a^{L}=\binom{1}{-u^{\varepsilon(N-2i)}},
\end{equation}
while maps leaving level $i$ toward a lower or higher level have respective
row lines
\begin{equation}\label{eq:local-row-lines}
 \bm{w}_a^{L}=\begin{pmatrix}1&u^{\varepsilon(2i-N)}\end{pmatrix},
 \qquad
 \bm{w}_a^{H}=\begin{pmatrix}1&u^{2\varepsilon i}\end{pmatrix}.
\end{equation}
These formulas follow by factoring \eqref{eq:lambda-peirce-block}.  The
image of that block is spanned by the first column
$(\gamma_{a-b},-\gamma_{-a-b})$, and its ratio is
$-\gamma_{-a-b}/\gamma_{a-b}=-u^{[-a-b]-[a-b]}$ by
\eqref{eq:cq-definition}.  Take $a=i$ with $1\leq i<n$ and $b$ of the
opposite parity.  For $\bar b<i$ the odd representatives are
$[a-b]=i-\bar b$ and $[-a-b]=N-i-\bar b$; for $i<\bar b<N-i$ they are
$[a-b]=N+i-\bar b$ and $[-a-b]=N-i-\bar b$; for $\bar b>N-i$ they are
$[a-b]=N+i-\bar b$ and $[-a-b]=2N-i-\bar b$.  The ratio is therefore
$-u^{-2i}$ in the middle case, where $\ell(b)>\ell(a)$, and $-u^{N-2i}$
in the two outer cases, where $\ell(b)<\ell(a)$, whatever the vertex $b$
in the other level.  The row space of the block $e_b\Lambda e_a$ leaving
$a$ is spanned by $(\gamma_{b-a},\gamma_{b+a})$, and the same three cases
give $\gamma_{b+a}/\gamma_{b-a}=u^{2i}$ for $\ell(b)>\ell(a)$ and
$u^{2i-N}$ for $\ell(b)<\ell(a)$.  Replacing $a$ by $-a$ interchanges
$[a-b]$ with $[-a-b]$ and $[b-a]$ with $[b+a]$, and so negates every
exponent, which is the case $\varepsilon=-1$.  The four pairings of the lines are
\begin{align}
 \bm{w}_a^{L}\bm{v}_a^{H}&=1-u^{-\varepsilon N}\neq0,&
 \bm{w}_a^{L}\bm{v}_a^{L}&=0,\label{eq:local-pairings-down}\\
 \bm{w}_a^{H}\bm{v}_a^{L}&=1-u^{\varepsilon N}\neq0,&
 \bm{w}_a^{H}\bm{v}_a^{H}&=0.\label{eq:local-pairings-up}
\end{align}
The inequalities are precisely where admissibility is used.  For the
endpoint maps we use
\[
 \bm{v}_0^H=\bm{v}_n^L=\binom{1}{-1},
 \qquad
 \bm{w}_0^H=\bm{w}_n^L=\begin{pmatrix}1&1\end{pmatrix},
\]
and the vanishing pairings $\bm{w}_0^H\bm{v}_0^H=0$ and $\bm{w}_n^L\bm{v}_n^L=0$ show that a
change of direction at the endpoint levels $0$ and $n$ is again zero.

For vertices on distinct levels, define rank-one maps
\begin{equation}\label{eq:Xab-definition}
 X_{a,b}=
 \begin{cases}
   \bm{v}_a^{H}\bm{w}_b^{L},&\ell(a)<\ell(b),\\
   \bm{v}_a^{L}\bm{w}_b^{H},&\ell(a)>\ell(b),
 \end{cases}
\end{equation}
using the endpoint lines just displayed.  For adjacent levels,
\eqref{eq:lambda-peirce-block} is a nonzero scalar multiple of
$X_{a,b}$, so $X_{a,b}\in A$.  Products along a strictly monotone chain
are nonzero by
\eqref{eq:local-pairings-down}--\eqref{eq:local-pairings-up}; hence they
produce $X_{a,b}$ for arbitrary
distinct endpoint levels.  Conversely, every nonzero
$e_a\Lambda e_b$ is a scalar multiple of the appropriate $X_{a,b}$.
Since $T=\sum_a\omega^ae_a$, and since the product rule of the next
paragraph shows the span below to be closed under multiplication, the
algebra generated by $T$ and $\Lambda$ is exactly the span of
\begin{equation}\label{eq:peirce-basis}
 \{e_a:a\in\ZZ/N\ZZ\}
 \cup
 \{X_{a,b}:\ell(a)\neq\ell(b)\}.
\end{equation}
The displayed elements lie in distinct Peirce blocks and are therefore
linearly independent.  This proves \eqref{eq:peirce-blocks}, and also
shows that every $e_a$ is primitive.

The pairings show that $X_{a,b}X_{b,c}$ is a nonzero scalar multiple of
$X_{a,c}$ exactly for a strictly monotone triple of levels, and is zero
after a change of direction.  At an interior source $b=\varepsilon i$,
divide a downward map $X_{a,b}$ by $1-u^{-\varepsilon N}$ and an upward
map by $1-u^{\varepsilon N}$.  Maps out of the two endpoint vertices are
left unchanged; these vertices never occur as intermediate vertices of a
strictly monotone path.  In every monotone product the intermediate factor
cancels, giving \eqref{eq:monotone-multiplication}.  The resulting
multiplication table contains no $u$, proving parameter independence.

Finally, let $J_+$ be the span of the $x_{a,b}$.  The multiplication rule
makes $J_+$ an ideal; it is nilpotent, because a nonzero product of
basis elements follows a strictly monotone sequence of levels in
$\{0,\ldots,n\}$.  Since $A/J_+\cong\CC^N$ is semisimple,
$J_+=\operatorname{Jac}(A)$.  A product $x_{a,b}x_{b,c}$ is nonzero only if
$|\ell(a)-\ell(c)|\geq2$, and every $x_{a,c}$ with $|\ell(a)-\ell(c)|\geq2$
factors through a vertex on an intermediate level.  Hence $J_+^2$ is
spanned by the $x_{a,b}$ with $|\ell(a)-\ell(b)|\geq2$, and
$e_a(J_+/J_+^2)e_b$ is nonzero exactly when the level distance is one, which
gives the arrows of the Gabriel quiver.  Under the stated relations every
path that changes level direction is zero, and all monotone paths with the
same endpoints coincide.  The path algebra modulo these relations is
therefore spanned by the trivial path at each vertex and by one monotone
path for each ordered pair of vertices on distinct levels.  Their images in $A$ are nonzero
scalar multiples of the elements \eqref{eq:peirce-basis}, which are
linearly independent, so no further relation holds.
\end{proof}

The presentation also gives a pullback description.  Put
\[
 S=\bigoplus_a\CC e_a,\qquad
 J_{\uparrow}=\Span\{x_{a,b}:\ell(a)>\ell(b)\},\qquad
 J_{\downarrow}=\Span\{x_{a,b}:\ell(a)<\ell(b)\}.
\]
The multiplication rule makes $J_{\uparrow}$ and $J_{\downarrow}$
two-sided ideals with zero mixed products, and
$A=S\oplus J_{\uparrow}\oplus J_{\downarrow}$ as a vector space.
The subalgebras $A_{\uparrow}=S\oplus J_{\uparrow}$ and
$A_{\downarrow}=S\oplus J_{\downarrow}$ project onto $S$ by deleting
their off-diagonal parts.  The map
\[
 h+p+q\longmapsto(h+p,h+q),\qquad
 h\in S,\quad p\in J_{\uparrow},\quad q\in J_{\downarrow},
\]
therefore identifies $A$ with
$A_{\uparrow}\times_S A_{\downarrow}$, the pullback of these two
projections.  This is the construction with both mixed-path products zero
described in \cite[Sec.~3, before Cor.~3.7]{XiXu2013}.  Each factor is an
incidence algebra \cite{Rota1964}.  Order $\ZZ/N\ZZ$ by declaring $a\preceq b$
when $a=b$ or $\ell(a)<\ell(b)$.  Then $e_a\mapsto\epsilon_{a,a}$ and
$x_{a,b}\mapsto\epsilon_{a,b}$ identify $A_{\downarrow}$ with the incidence
algebra of $(\ZZ/N\ZZ,\preceq)$, in which
$\epsilon_{a,b}\epsilon_{b',c}=\delta_{b,b'}\epsilon_{a,c}$; likewise
$A_{\uparrow}$ is the incidence algebra of the opposite order.

\begin{remark}[The algebra is not monomial]\label{rem:not-monomial}
Let $a$ be the vertex of level $0$, let $b,b'$ be the two vertices of level
$1$, and let $c$ be a vertex of level $2$.  Since
$e_bAe_c=\CC x_{b,c}$ and $e_b\operatorname{Jac}(A)^2e_c=0$, every
presentation of $A$ by its
Gabriel quiver sends the arrow from $c$ to $b$ to a nonzero multiple of
$x_{b,c}$, and similarly for the other three arrows involved.  The two
paths of length two from $c$ to $a$ therefore map to nonzero multiples of
$x_{a,b}x_{b,c}=x_{a,c}=x_{a,b'}x_{b',c}$.  In a monomial algebra, the
paths that are nonzero are linearly independent.  Hence $A_{N,u}$ is not
monomial for any even $N\geq4$; its relations include commutativity
relations.
\end{remark}

\begin{corollary}[Full radical filtration]\label{prop:radical}
For $1\leq r\leq n$,
\begin{equation}\label{eq:radical-power-basis}
 J^r=\bigoplus_{|\ell(a)-\ell(b)|\geq r}\CC x_{a,b}.
\end{equation}
In particular,
\begin{align}
 \dim A&=N^2-N+2,\label{eq:algebra-dimension}\\
 \dim J^r&=2+4(n-r)(n-r+1)\qquad(1\leq r\leq n),
 \label{eq:radical-power-dimension}\\
 \dim A/J&=N,\qquad
 \dim J^r/J^{r+1}=8(n-r)\quad(1\leq r<n),
 \qquad \dim J^n=2.\label{eq:loewy-layer-dimensions}
\end{align}
Moreover,
\begin{equation}\label{eq:radical-result}
 J=\operatorname{Jac}(A),\qquad J^n\neq0,\qquad J^{n+1}=0.
\end{equation}
Thus the Loewy length of $A$, the least $r$ for which $J^r=0$, is $n+1$.
For each $a$, the indecomposable right projective $e_aA$ has simple top
$S_a=e_aA/e_aJ$.  For $1\leq r\leq n$, its radical layer
$e_aJ^r/e_aJ^{r+1}$ contains once each simple $S_b$ with
$|\ell(a)-\ell(b)|=r$, and no other simple.
\end{corollary}

\begin{proof}
Let $J_+$ be the span of the $x_{a,b}$.  By the last paragraph of the proof
of Theorem~\ref{thm:monotone-algebra}, $J_+=\operatorname{Jac}(A)$.
It contains $\Lambda$ by \eqref{eq:lambda-peirce-block}.  Conversely, all
adjacent $x_{a,b}$ lie in $A\Lambda A$, and every other $x_{a,b}$ is a
monotone product of adjacent ones.  Thus $J_+=A\Lambda A=J$.

Expand a product of $r$ radical elements in the basis $\{x_{a,b}\}$.
Every nonzero summand has a strictly monotone level sequence, so its
endpoint levels differ by at least $r$.  Conversely, if
two levels differ by at least $r$, inserting $r-1$ strictly intermediate
levels gives a nonzero product equal to the endpoint map.  This proves
\eqref{eq:radical-power-basis}.

Put $\mu_0=\mu_n=1$ and $\mu_i=2$ for $1\leq i<n$.  The number of ordered
pairs on different levels is
\[
 N^2-\sum_{i=0}^n\mu_i^2=N^2-2N+2.
\]
Adding the $N$ diagonal idempotents proves
\eqref{eq:algebra-dimension}.  For $1\leq r<n$, the number of ordered
vertex pairs at exact level distance $r$ is
\[
 2\sum_{i=0}^{n-r}\mu_i \mu_{i+r}=8(n-r),
\]
while the two endpoint pairs are the only pairs at distance $n$.  Summing
the exact-distance counts from $r$ through $n$ proves
\eqref{eq:radical-power-dimension}--\eqref{eq:loewy-layer-dimensions}.
Equation
\eqref{eq:radical-result} follows immediately.  Left multiplication by
$e_a$ in \eqref{eq:radical-power-basis} gives the projective layers.
\end{proof}

\begin{corollary}[Center, Cartan matrix, and representation type]
\label{cor:representation-type}
The center of $A$ is $\CC$.  With
$C_{a,b}=\dim_{\CC}e_aAe_b$, its Cartan matrix is
\begin{equation}\label{eq:cartan-matrix}
 C_{a,b}=
 \begin{cases}
  1,&a=b\text{ or }\ell(a)\neq\ell(b),\\
  0,&a\neq b\text{ and }\ell(a)=\ell(b).
 \end{cases}
\end{equation}
It has rank $N-1$ and kernel spanned by $\bm{\delta}_0-\bm{\delta}_n$, where
$\bm{\delta}_a$ is the coordinate vector indexed by $a$.  Consequently,
\begin{equation}\label{eq:infinite-global-dimension}
 \operatorname{gldim}A=\infty,
\end{equation}
and $A$ has infinite representation type.  In fact, $A$ has a quotient
isomorphic to the path algebra of the quiver with two sources $p_1,p_2$, two
sinks $q_1,q_2$, and one arrow $p_i\to q_j$ for every $i,j\in\{1,2\}$, a
Euclidean quiver of type $\widetilde{\mathbb A}_3$.
\end{corollary}

\begin{proof}
Let $Z$ be a central element.  Commutation with the primitive idempotents
gives $e_aZe_b=0$ for $a\neq b$, while
$e_aAe_a=\CC e_a$.  Hence $Z=\sum_a\nu_ae_a$.  For every pair of
vertices in consecutive levels, centrality and $x_{a,b}\neq0$ give
\[
 \nu_ax_{a,b}=Zx_{a,b}=x_{a,b}Z=\nu_bx_{a,b}.
\]
The Gabriel quiver is connected, so all $\nu_a$ agree.  Thus
$Z(A)=\CC$.

Formula \eqref{eq:cartan-matrix} follows from
\eqref{eq:peirce-blocks}.  To compute its kernel, take $\bm h=(h_a)_a$
with $C\bm h=0$ and put
\[
 h_\Sigma=\sum_a h_a,
 \qquad g_i=\sum_{\ell(a)=i}h_a,
 \qquad \mu_i=\#\{a:\ell(a)=i\}.
\]
The row equation at a vertex $a$ of level $i$ is
$h_a+h_\Sigma-g_i=0$.  Thus the coordinates are constant on each level, and
summation over that level yields
\[
 (1-\mu_i)g_i=-\mu_i h_\Sigma.
\]
Since $\mu_0=\mu_n=1$, the endpoint equations give $h_\Sigma=0$.  Since
$\mu_i=2$ for $0<i<n$, every interior coordinate vanishes, and the two
endpoint coordinates sum to zero.  This proves the rank and kernel claims.

If $A$ had finite global dimension, every simple module would have a finite
projective resolution.  In the Grothendieck group, the simple classes would
then be integral linear combinations of the indecomposable-projective
classes, making the Cartan matrix invertible over $\ZZ$; this is
Eilenberg's observation that the Cartan determinant of an algebra of
finite global dimension is $\pm1$ \cite{Eilenberg1954}.  The computed
kernel rules this out and proves \eqref{eq:infinite-global-dimension}.

It remains to construct the quotient.  If $n\geq3$, retain the two vertices
at level $1$, the two vertices at level $2$, and the four arrows directed
from level $1$ to level $2$.  Quotient by all other vertices and arrows.  If
$n=2$, retain the two endpoints as sources, the two level-$1$ vertices as
sinks, and the four arrows directed from the endpoints to level $1$; quotient
by the reverse arrows.  In the latter case the two endpoint maps also vanish
in the quotient, since every representing length-two path uses a discarded
reverse arrow.  The presentation in Theorem~\ref{thm:monotone-algebra} shows
in both cases that the quotient is precisely the stated four-arrow path
algebra.  Its underlying graph is not a Dynkin diagram, so the quotient has
infinite representation type by Gabriel's theorem
\cite{Gabriel1972Unzerlegbare,Bernstein_1973cf}.  An explicit family is as
follows.

For $\lambda\in\CC^\times$, put a one-dimensional vector space at every
vertex of this quotient and assign the arrow scalars
\[
 \alpha_{11}=\alpha_{12}=\alpha_{21}=1,
 \qquad \alpha_{22}=\lambda.
\]
Every endomorphism is scalar because all four arrow maps are nonzero, so the
module is indecomposable.  Under changes of basis at the four vertices, the
cross-ratio
\[
 \frac{\alpha_{11}\alpha_{22}}{\alpha_{12}\alpha_{21}}=\lambda
\]
is invariant.  Distinct values of $\lambda$ therefore give pairwise
nonisomorphic indecomposable modules.  Inflating them along the quotient map
proves that $A$ has infinite representation type.
\end{proof}

\begin{proof}[Proof of Theorem~\ref{thm:main}]
Proposition~\ref{prop:vanishing} gives vanishing for every $M\geq n$.
The all-ones certificate \eqref{eq:certificate-identity} and
Lemma~\ref{lem:abstract-threshold} show that vanishing cannot occur
earlier, so $\wtau_N(u)=N/2$.
Proposition~\ref{prop:sharp-classification} supplies the sharp-layer
criterion and the rank-two statement.
Theorem~\ref{thm:monotone-algebra} and Corollaries~\ref{prop:radical}
and~\ref{cor:representation-type} give the remaining structural assertions.
\end{proof}

\section{An arbitrary-field filtered extension}\label{sec:axiomatic-extension}

The preceding proofs use only a small amount of structure.  We record it
explicitly, both to separate the mechanism from the complex root-of-unity
realization and to make clear which hypotheses are responsible for each
conclusion.

\begin{definition}[Chebyshev-filtered datum]\label{def:chebyshev-datum}
Let $\Bbbk$ be a field, let $E$ be finite-dimensional over $\Bbbk$, let
$n\ge2$, and let $U,\sigma\in\operatorname{GL}(E)$.  Set
\[
 C=U+U^{-1},\qquad D=\sigma(U-U^{-1}),
\]
and, on $E\oplus E$, define
\[
 \Theta=\begin{pmatrix}U^{-1}&0\\ U-U^{-1}&U\end{pmatrix},\qquad
 \Xi=\begin{pmatrix}0&\sigma\\0&0\end{pmatrix}.
\]
Let $\mathcal G\subseteq\operatorname{End}(E\oplus E)$ be the unital
subalgebra generated by $\Theta$ and $\Xi$, and put
$\mathcal J=\mathcal G\Xi\mathcal G$.  We call $(E,U,\sigma,F_\bullet)$ a length-$n$ Chebyshev-filtered datum if
there is a flag $0=F_0\subset\cdots\subset F_n=E$ such that
\begin{enumerate}[label=(H\arabic*),leftmargin=*,itemsep=1pt]
\item $CF_r\subseteq F_r$;
\item $DF_r\subseteq F_{r-1}$;
\item $D^{\,n-1}\ne0$;
\item $\Bbbk[\Theta,\Theta^{-1}]$ is semisimple.
\end{enumerate}
Here the invariance and lowering conditions hold for $1\leq r\leq n$.
No dimension is prescribed for the successive flag quotients.
\end{definition}

With $\Pi=\left(\begin{smallmatrix}I&0\\-I&I\end{smallmatrix}\right)$ one has
$\Theta=\Pi\operatorname{diag}(U^{-1},U)\Pi^{-1}$.  The minimal polynomial
of $\operatorname{diag}(U^{-1},U)$ is the least common multiple of those of
$U^{-1}$ and $U$, and $\Bbbk[U^{-1}]=\Bbbk[U]$.  Hence (H4) holds if and only if
$\Bbbk[U]$ is semisimple, that is, if and only if the minimal polynomial of
$U$ has no repeated irreducible factor.

\begin{lemma}[Block powers]\label{lem:axiomatic-blocks}
Let $q_k$ be the integer polynomials determined by $q_0=0$, $q_1=1$, and
$q_{k+1}(x)=xq_k(x)-q_{k-1}(x)$ (extended to negative $k$ by the same
recurrence).  Then
\[
 U^k-U^{-k}=(U-U^{-1})q_k(C),\qquad
 \Theta^k=\begin{pmatrix}U^{-k}&0\\U^k-U^{-k}&U^k\end{pmatrix},
\]
and
\begin{equation}\label{eq:axiomatic-sandwich}
 \Xi\Theta^k\Xi=
 \begin{pmatrix}0&Dq_k(C)\sigma\\0&0\end{pmatrix}.
\end{equation}
\end{lemma}

\begin{proof}
The first identity follows from the recurrence and the fact that $U$ commutes
with $U^{-1}$.  The block formula is immediate for $k=0,1$ and is preserved
by multiplication by $\Theta$; multiplication by the explicit inverse of
$\Theta$ gives negative $k$.  Equation~\eqref{eq:axiomatic-sandwich} follows
by direct block multiplication.
\end{proof}

\begin{theorem}[Axiomatic threshold and radical]\label{thm:axiomatic-main}
For every length-$n$ Chebyshev-filtered datum and every integer tuple,
\[
 \Xi\Theta^{k_1}\Xi\cdots\Xi\Theta^{k_M}\Xi=0\quad(M\ge n),
\]
whereas the all-ones word at $M=n-1$ is nonzero.  Moreover,
\[
 \mathcal J=\operatorname{Jac}\mathcal G,\qquad
 \mathcal J^{\,n}\ne0,\qquad \mathcal J^{\,n+1}=0,
\]
so $\mathcal G$ has exactly $n+1$ nonzero Loewy layers.
\end{theorem}

\begin{proof}
By Lemma~\ref{lem:axiomatic-blocks}, the only possibly nonzero block of a
word is
\[
 Dq_{k_1}(C)Dq_{k_2}(C)\cdots Dq_{k_M}(C)\sigma.
\]
Conditions (H1) and (H2) imply that each factor $Dq_k(C)$ lowers the flag
by one, so $M\ge n$ gives zero.  Since $q_1=1$, the all-ones word has upper
right block $D^{n-1}\sigma\ne0$ by (H3) and invertibility of $\sigma$.
Every element of $\mathcal G$ is a linear combination of words in $\Theta$
and $\Xi$.  Expanding a product of $n+1$ elements of $\mathcal J$ therefore
gives a sum of words containing at least $n+1$ copies of $\Xi$.  Combining
each intervening string of powers of $\Theta$ reduces these to the vanishing
words above (including zero exponents).  Hence $\mathcal J^{n+1}=0$,
whereas $(\Xi\Theta)^{n-1}\Xi\in\mathcal J^n$ is nonzero.
Since $\Theta$ is invertible on a finite-dimensional space, its inverse is
a polynomial in $\Theta$.  Thus $\mathcal G/\mathcal J$ is a
quotient of the semisimple algebra $\Bbbk[\Theta,\Theta^{-1}]$ by (H4), so
the standard radical criterion gives $\mathcal J=\operatorname{Jac}\mathcal G$.
\end{proof}

\begin{proposition}[The root-of-unity model is an instance]\label{prop:axiomatic-instance}
For $N=2n$ and admissible $u$, the cyclic model of Sections~\ref{sec:reduction}
and~\ref{sec:flag} gives a length-$n$ Chebyshev-filtered datum with
$E=\CC^N$, $U=t$, $\sigma=s$, and $F_r=\cK_r$.  With
$\Gamma=\left(\begin{smallmatrix}I&I\\-I&0\end{smallmatrix}\right)$, the
change of basis used in \cite[Sec.~5]{QiuGuanYu2026Kaleidoscope}, one has
$\Theta=\Gamma^{-1}T\Gamma$ and $\Xi=\Gamma^{-1}\Lambda\Gamma$; thus
Theorem~\ref{thm:axiomatic-main} recovers the chained-word threshold and radical
nilpotency already proved above.
\end{proposition}

\begin{proof}
In the cyclic model, $C=t+t^{-1}$ and $D=s(t-t^{-1})=B_1$.  By
\eqref{eq:t-action} and \eqref{eq:shifted-x-rho}, applying $t^{\pm1}$ to
$P(x)+\rho Q(x)$ replaces the argument $x$ by $c_1x\pm\eta_1y$ and multiplies
the $\rho$-sector by $-1$.  The sum of the two substitutions is even in $y$,
hence, after $y^2=x^2-4$, again a polynomial in $x$ of no larger degree on
each parity sector.  Thus $C\cK_r\subseteq\cK_r$, which is (H1).  The lowering
condition (H2) is Proposition~\ref{prop:lowering} with $k=1$, and the
exhaustive flag is supplied by Lemma~\ref{lem:parity-basis}.
The survival condition (H3) is the nonzero
sharp certificate \eqref{eq:certificate-identity}.  Since $t^N=\Id$, the
minimal polynomial of $t$ divides $X^N-1$ and is square-free over $\CC$, so
(H4) holds.  Direct multiplication with
$\Gamma^{-1}=\left(\begin{smallmatrix}0&-I\\I&I\end{smallmatrix}\right)$
gives the asserted block forms.
\end{proof}

\begin{theorem}[Independence of the four conditions]\label{thm:axiom-independence}
Each of (H1)--(H4) is essential: for every $i$ there is a finite-dimensional
rational datum with an exhaustive flag satisfying the other three conditions
for which a conclusion of
Theorem~\ref{thm:axiomatic-main} fails.
\end{theorem}

\begin{proof}
All witnesses are over $\mathbb Q$, with $F_0=0$ and $F_n=E$.

For (H1), take $n=3$, $E=\mathbb Q^3$, $U=\operatorname{diag}(1,2,3)$,
\[
\sigma=\begin{pmatrix}0&2/3&0\\1&-2/3&3/8\\-1&-2/3&3/8\end{pmatrix},
\qquad F_1=\langle \bm{e}_1\rangle,\quad F_2=\langle \bm{e}_1,\bm{e}_2+\bm{e}_3\rangle.
\]
Writing $\bm v=\bm{e}_2+\bm{e}_3$, one has $D\bm{e}_1=0$,
$D\bm{e}_2=\bm{e}_1-\bm v$, and $D\bm{e}_3=\bm v$.
Thus $DF_2\subseteq F_1$, $DE\subseteq F_2$, and $D^2\bm{e}_3=\bm{e}_1\ne0$.
However, $C=\operatorname{diag}(2,5/2,10/3)$ does not preserve $F_2$.
The induced action of $DC$ on the line $F_2/F_1$ is multiplication by
$5/6$, and $DC\bm{e}_3=(10/3)\bm v$, so $(DC)^3\bm{e}_3\ne0$.
Since $q_2(C)=C$, the all-twos word at $M=3$ is nonzero.
The matrix $\sigma$ is invertible.  Its first row forces the second
coordinate of a kernel vector to vanish, and subtracting the other
two row equations then forces the first and third to vanish.
Condition (H4) holds because $U$ is diagonal.

For (H2), take $n=2$, $E=\mathbb Q^2$, $U=\operatorname{diag}(2,3)$,
$\sigma=I$, and $F_1=\langle \bm{e}_1\rangle$.  Then $C$ preserves the flag,
and $D=\operatorname{diag}(3/2,8/3)$ is invertible, so (H3) and (H4) hold
while (H2) fails.  The all-ones word at $M=2$ has nonzero block
$D^2\sigma$, so it survives.

For (H3), take $n=3$, $E=\mathbb Q^3$, $U=\sigma=I$, and
$F_r=\langle \bm{e}_1,\ldots,\bm{e}_r\rangle$.  Then $C=2I$, $D=0$, and
$\Theta=I$, so the other conditions hold.  The algebra is
$\mathcal G=\mathbb QI\oplus\mathbb Q\Xi$ with $\Xi^2=0\ne\Xi$; its Loewy
length is two, not four, and the claimed sharp word is zero.

For (H4), take $n=2$, $E=\mathbb Q^2$, $U=I+\mathcal N$ with
$\mathcal N=\left(\begin{smallmatrix}0&1\\0&0\end{smallmatrix}\right)$,
$\sigma=\tfrac12I$, and $F_1=\ker\mathcal N$.  Here $C=2I$ and
$D=\mathcal N\ne0$, so (H1)--(H3) hold, but $U$ is not semisimple and
the minimal polynomial of $\Theta$ is $(x-1)^2$.  Writing $H=\Theta-I$ gives
$H^2=\Xi^2=H\Xi H=0$ and
\[
 \mathcal J=\operatorname{span}_{\mathbb Q}\{\Xi,H\Xi,\Xi H,\Xi H\Xi\},
 \qquad H\notin\mathcal J.
\]
Indeed, $H=\left(\begin{smallmatrix}-\mathcal N&0\\2\mathcal N&\mathcal N\end{smallmatrix}\right)$
has nonzero lower-left block, whereas every displayed generator of
$\mathcal J$ has zero lower-left block.  The quotient is generated by the
nonzero square-zero class of $H$, so
$\operatorname{Jac}(\mathcal G/\mathcal J)\ne0$ and radical identification
fails.  These witnesses prove the assertion.
\end{proof}

\begin{remark}[Flag dimensions and exhaustivity]
The proof of Theorem~\ref{thm:axiomatic-main} uses no flag-dimension formula.
In particular, the dimensions $\dim\cK_r=2r$ of the cyclic model are a
feature of that realization, not a necessary hypothesis of the abstract
threshold theorem.  Exhaustivity is different: if $F_n\ne E$, lowering
on the displayed subspaces gives no control on vectors outside $F_n$.
\end{remark}

\section{Concluding remarks}

The termwise conjecture of Qiu, Guan, and Yu has an exact structural answer.
A single flag controls all exponent tuples and identifies the products that
survive at the last possible level.  The spectral blocks determine how
shorter products multiply, giving a parameter-independent algebra, a pullback
of two incidence algebras, together with its full radical filtration.  The
threshold $N/2$ and the nilpotency index $N/2+1$ of the radical are two
descriptions of the same fact: $J^{N/2}$ is the last nonzero radical power.
The functional-calculus extension separates the two parity sectors and
produces rank-one elements of $J^{N/2}$ for every even $N\geq6$.
Corollary~\ref{cor:representation-type} shows that the algebra is
representation-infinite; whether it is tame or wild in the sense of
Drozd's dichotomy \cite{Drozd_1980ta} is not decided here.

The arbitrary-field theorem isolates the preservation and lowering
conditions behind the vanishing argument, together with the hypotheses
needed for sharpness and radical identification.  The parity decomposition
and the Fourier coefficients of $s$ are additional structure of the
even-order realization.  Extending the explicit presentation to other
realizations, including odd orders $N$, requires replacements for these ingredients.  It would also be useful to
relate the radical layers to scattering quantities in the kaleidoscope
model; the presentation obtained here supplies the algebraic data for that
question.

\section*{Statements and Declarations}

\subsection*{Funding}

\fundingstatement

\subsection*{Competing interests}

The authors have no relevant financial or non-financial interests to disclose.

\subsection*{Data availability}

Data sharing is not applicable to this article, as no datasets were
generated or analysed during the current study.  All results are
established by the proofs and explicit constructions in the paper.

\subsection*{Use of large language models}

OpenAI GPT-5.6 Sol, OpenAI GPT-6 Astra, Anthropic Claude Opus~5.5, and
Anthropic Claude Fable~5.1 were used throughout this work: to develop and
check the arguments, to review the proofs and their hypotheses, and to
organize and typeset the manuscript in LaTeX.  The authors take full
responsibility for the content of this paper.

\bibliography{references}

\end{document}